\documentclass[envcountsame,orivec,runningheads,a4paper]{llncs}

\usepackage{amsmath,amssymb,amsthm}
\usepackage{enumerate}
\usepackage{hyperref}

\usepackage{lmodern}
\DeclareMathAlphabet{\mathpzc}{OT1}{pzc}{m}{it}
\usepackage[T1]{fontenc}
\usepackage[utf8]{inputenc}

\usepackage[textsize=tiny]{todonotes}

\DeclareMathOperator{\pref}{pref}
\DeclareMathOperator{\suff}{suff}
\providecommand{\abs}[1]{\lvert#1\rvert}
\newcommand{\infw}[1]{\mathbf{#1}}

\newcommand{\act}[2][]{\mathpzc{A\mkern-3mu c}_{#1}(#2)}
\newcommand{\abexp}[2]{\mathpzc{A\mkern-3mu e}_{#1}(#2)}

\begin{document}
  \title{Every nonnegative real number is an abelian critical exponent}
  \author{Jarkko Peltomäki\inst{1,2,3}\orcidID{0000-0003-3164-1559} \and Markus A. Whiteland\inst{3}}
  \institute{
    The Turku Collegium for Science and Medicine TCSM, University of Turku, Turku, Finland \and
    Turku Centre for Computer Science TUCS, Turku, Finland \and
    University of Turku, Department of Mathematics and Statistics, Turku, Finland\\
    \email{\{jspelt,mawhit\}@utu.fi}
  }

  \date{}
  \maketitle

  \begin{abstract}
    The abelian critical exponent of an infinite word $w$ is defined as the maximum ratio between the exponent and
    the period of an abelian power occurring in $w$. It was shown by Fici et al. that the set of finite abelian
    critical exponents of Sturmian words coincides with the Lagrange spectrum. This spectrum contains every large
    enough positive real number. We construct words whose abelian critical exponents fill the remaining gaps, that is,
    we prove that for each nonnegative real number $\theta$ there exists an infinite word having abelian critical
    exponent $\theta$. We also extend this result to the $k$-abelian setting.

    \keywords{abelian equivalence \and $k$-abelian equivalence \and critical exponent \and Sturmian word}
  \end{abstract}

  \section{Introduction}
  The study of powers and their avoidance has been one of the central themes in combinatorics on words; see
  \cite[Ch.~4]{2016:combinatorics_words_and_symbolic_dynamics}. The central notion here is that of the \emph{critical
  exponent} which measures the maximum exponent of a power occurring in a given word. Recently it has been popular to
  generalize the notion of a power using some equivalence relation in place of the usual equality of words. For
  example, abelian equivalence (see the references of \cite{2016:abelian_powers_and_repetitions_in_sturmian_words}),
  and its generalizations $k$-abelian equivalence
  \cite{2013:on_a_generalization_of_abelian_equivalence_and_complexity_of_infinite,2017:on_growth_and_fluctuation_of_k_abelian_complexity}
  and binomial equivalence
  \cite{2015:another_generalization_of_abelian_equivalence_binomial,2015:avoiding_2-binomial_squares_and_cubes} have
  been popular options.

  Two words $u$ and $v$ are \emph{abelian equivalent}, written $u \sim v$, if they are permutations of each other. An
  \emph{abelian power} of exponent $e$ and period $m$ is a word of the form $u_0 \cdots u_{e-1}$ such that
  $m = \abs{u_0}$ and $u_0$, $\ldots$, $u_{e-1}$ are nonempty and abelian equivalent. For example, $01 \cdot 10$ (a
  square) and $abc \cdot bca \cdot cab$ (a cube) are abelian powers. Now it is possible to define the abelian critical
  exponent of an infinite word as the maximum exponent of an abelian power occurring in it. However, this does not give
  any interesting information on abelian powers occurring in Sturmian words or, more generally, in words with bounded
  abelian complexity because such words contain abelian powers of arbitrarily high exponent
  \cite{2011:abelian_complexity_of_minimal_subshifts}. In order to capture more information on abelian powers of an
  infinite word to a single quantity, it was proposed in \cite{2016:abelian_powers_and_repetitions_in_sturmian_words}
  to define the \emph{abelian critical exponent} $\act{\infw{w}}$ of an infinite word $\infw{w}$ as the quantity
  \begin{equation*}
    \limsup_{m \to \infty} \frac{\abexp{\infw{w}}{m}}{m},
  \end{equation*}
  where $\abexp{\infw{w}}{m}$ is the supremum of exponents of abelian powers of period $m$ occurring in $\infw{w}$.
  This notion turns out to be much more interesting. For example, $\act{\infw{f}} = \sqrt{5}$ for the Fibonacci word
  $\infw{f}$, the fixed point of the substitution $0 \mapsto 01$, $1 \mapsto 0$
  \cite[Thm.~5.14]{2016:abelian_powers_and_repetitions_in_sturmian_words}. Furthermore $\sqrt{5}$ is the minimum
  abelian critical exponent among all Sturmian words
  \cite[Thm.~5.14]{2016:abelian_powers_and_repetitions_in_sturmian_words}. It follows that for each Sturmian word
  $\infw{s}$ and each $\delta > 0$, there exists an increasing sequence $(m_i)$ of integers such that $\infw{s}$
  contains an abelian power of period $m_i$ and total length greater than $(\sqrt{5} - \delta)m_i^2$. Notice that if
  $\infw{w}$ does not contain abelian powers with arbitrarily large exponent, then $\act{\infw{w}} = 0$. Many examples
  of such words are known; see, e.g., \cite[Ch.~4.6]{2016:combinatorics_words_and_symbolic_dynamics}. It is also
  possible that $\act{\infw{w}} = \infty$. Take for example the Thue-Morse word $\infw{t}$, the fixed point of the
  substitution $0 \mapsto 01$, $1 \mapsto 10$. Indeed, it is straightforward to see that $\infw{t}$ can be factored as
  a product of abelian equivalent words of length $2n$ for all $n \geq 0$. This shows that $\act{\infw{t}} = \infty$.

  Further study in \cite{2016:abelian_powers_and_repetitions_in_sturmian_words} showed the surprising fact that the set
  of finite abelian critical exponents of Sturmian words equals the Lagrange spectrum $\mathcal{L}$. The Lagrange
  constant of an irrational $\alpha$ is the infimum of the real numbers $\lambda$ such that for every $c > \lambda$ the
  inequality $\abs{\alpha - n/m} < 1/cm^2$ has only finitely many rational solutions $n/m$. The Lagrange spectrum is
  the set of finite Lagrange constants of irrational numbers. The Lagrange spectrum has been extensively studied in
  number theory since the works of Markov
  \cite{1879:sur_les_formes_quadratiques_binaires_indefinies,1880:sur_les_formes_quadratiques_binaires_indefinies_ii}
  in the 19th century. The famous theorems of Markov show that the initial part of $\mathcal{L}$ inside the interval
  $[\sqrt{5}, 3)$ is discrete. Later in 1947 Hall proved that $\mathcal{L}$ contains a half-line
  \cite{1947:on_the_sum_and_products_of_continued_fractions}. After a series of improvements by multiple authors, it
  was finally determined by Freiman in 1975 \cite{1975:diophantine_approximation_and_geometry_of_numbers} that the
  largest half-line contained in the Lagrange spectrum is $[c_F, \infty)$, where
  \begin{equation*}
    c_F = \frac{2221564096+283748\sqrt{462}}{491993569} = 4.5278295661 \ldots
  \end{equation*}
  Good sources for information on the Lagrange spectrum are the monograph of Cusick and Flahive
  \cite{1989:the_markoff_and_lagrange_spectra} and Aigner's book
  \cite{2013:markovs_theorem_and_100_years_of_the_uniqueness}. See also the recent book
  \cite{2019:from_christoffel_words_to_markoff_numbers} of Reutenauer for a more word-combinatorial flavor.

  The connection between the Lagrange spectrum and abelian critical exponents of Sturmian words shows that each
  real number larger than $c_F$ is the abelian critical exponent of some infinite word. This raises the obvious
  question of whether this can be extended to hold for all nonnegative numbers. In this paper, we answer the question
  in the positive. The main result of this paper is the following theorem.

  \begin{theorem}\label{thm:main}
    Let $\theta$ be a nonnegative real number. Then there exists an infinite word $\infw{w}$ such that
    $\act{\infw{w}} = \theta$. The word $\infw{w}$ can be taken over an alphabet of at most three letters.
  \end{theorem}

  This result should be compared with a result of Krieger and Shallit stating that every real number $\theta > 1$ is a
  critical exponent (in the usual sense) of some infinite word \cite{2007:every_real_number_greater_than_1}. Notice
  that here the number of letters required tends to infinity when $\theta$ tends to $1$
  \cite{2007:every_real_number_greater_than_1}, but in our setting we need at most three letters.

  We prove an analogue of \autoref{thm:main} for $k$-abelian critical exponents; see \autoref{sec:k-abelian} for the
  extension and the necessary definitions.

  Our proof method is to exploit the properties of the Lagrange spectrum, that is, the fact that \autoref{thm:main} is
  already known to be true for all reals greater than $c_F$. The idea is to find a suitable $N$-uniform substitution
  $\sigma$ such that each abelian power in $\sigma(\infw{w})$ can be decoded to an abelian power in $\infw{w}$ with the
  same exponent. This means, in essence, that the abelian powers in $\sigma(\infw{w})$ are the abelian powers of
  $\infw{w}$ blown up by a factor of $N$. Roughly speaking, the ratio of exponents and periods corresponding to
  $\act{\infw{w}}$ gets divided by $N$, that is, $\act{\infw{\sigma(\infw{w})}} = \act{\infw{w}}/N$. The conclusion is
  that \autoref{thm:main} is true for each real in the interval $[c_F/N, \infty)$, where $[c_F, \infty)$ is the largest
  half-line contained in the Lagrange spectrum. We may choose $N$ to be arbitrarily large, so \autoref{thm:main}
  follows. The extension of \autoref{thm:main} to the $k$-abelian setting is proved using the same ideas.

  We use the usual notions and notation from combinatorics on words. If the reader encounters anything undefined, we
  refer him or her to \cite{2002:algebraic_combinatorics_on_words}. Even though we mention Sturmian words several times
  in this paper, we do not need any properties of these binary words. For their definition, we refer the reader to
  \cite[Ch.~2]{2002:algebraic_combinatorics_on_words} and \cite[Ch.~4]{diss:jarkko_peltomaki}.
  
  \section{Proof of \autoref{thm:main}}\label{sec:abelian}
  Let $\theta$ be a nonnegative real number. If $\theta = 0$, then $\theta$ is the abelian critical exponent of any
  infinite word that avoids abelian powers with large enough exponent. Such words exist by
  \cite{1979:strongly_non-repetitive_sequences_and_progression-free_sets} (abelian fourth powers are avoidable over two
  letters); see also \cite[Ch.~4.6]{2016:combinatorics_words_and_symbolic_dynamics}.

  Assume then that $\theta > 0$, and let $N$ be an integer such that $N\theta \in [c_F, \infty)$. Let $\infw{w}$ be an
  infinite binary word. Our aim is to find an $N$-uniform substitution $f$ defined on a two-letter alphabet with the
  following properties:
  \begin{enumerate}[(i)]
    \item If an abelian power $u_0 \cdots u_{e-1}$ occurs in $\infw{w}$, then $f(u_0) \cdots f(u_{e-1})$ is an abelian
          power occurring in $f(\infw{w})$.
    \item If an abelian power $u_0 \cdots u_{e-1}$, $e \geq N$, occurs in $f(\infw{w})$, then $\infw{w}$ contains an
          abelian power $v_0 \cdots v_{e-1}$ with $\abs{v_0} = \abs{u_0}/N$.
  \end{enumerate}
  Let us show how to prove \autoref{thm:main} under the assumption that such $f$ exists.

  Let $\infw{s}$ be a Sturmian word having $\act{\infw{s}} = N\theta$. In fact, any
  binary word $\infw{s}$ with $\act{\infw{s}} = N\theta$ will do, we just know that such a Sturmian word exists by
  the results of \cite{2016:abelian_powers_and_repetitions_in_sturmian_words}. We claim that
  $\act{f(\infw{s})} = \theta$. This proves \autoref{thm:main} when $\theta > 0$ (assuming that $f(\infw{w})$ has at
  most three letters).

  By Property (i), we have $\abexp{f(\infw{s})}{tN} \geq \abexp{\infw{s}}{t}$ for all positive integers $t$. 
  Since $\act{\infw{s}} > 0$, the word $\infw{s}$ contains abelian powers of arbitrarily high exponent, and thus
  by Property (i) the word  $f(\infw{s})$ contains abelian powers of arbitrarily high exponent and period divisible by
  $N$. If $\abexp{f(\infw{s})}{tN} \geq N$, then $\abexp{f(\infw{s})}{tN} \leq \abexp{\infw{s}}{t}$ by Property (ii).
  Therefore there exists a sequence $(t_i)$ such that $\abexp{\infw{s}}{t_i} = \abexp{f(\infw{s})}{t_i N}$ for all $i$.
  Hence
  \begin{equation*}
    \limsup_{i \to \infty} \frac{\abexp{f(\infw{s})}{t_i N}}{t_i N} = \limsup_{i \to \infty} \frac{\abexp{\infw{s}}{t_i}}{t_i N} = \frac{1}{N} \limsup_{i \to \infty} \frac{\abexp{\infw{s}}{t_i}}{t_i} = \frac{1}{N} \act{\infw{s}} = \theta,
  \end{equation*}
  so $\act{f(\infw{s})} \geq \theta$. If $\act{f(\infw{s})} > \theta$, then there exists an increasing sequence
  $(\ell_i)$ such that
  \begin{equation*}
    \frac{\abexp{f(\infw{s})}{\ell_i}}{\ell_i} > \theta > 0
  \end{equation*}
  for all $i$. By the preceding, only finitely many of the numbers in the sequence $(\ell_i)$ are divisible by $N$. By
  Property (ii), we thus have $\abexp{f(\infw{s})}{\ell_i} \geq N$ only for finitely many $i$ meaning that
  \begin{equation*}
    \frac{\abexp{f(\infw{s})}{\ell_i}}{\ell_i} < \frac{N}{\ell_i}
  \end{equation*}
  for $i$ large enough. This is impossible as $N/\ell_i \to 0$ as $i \to \infty$. The conclusion is that
  $\act{f(\infw{s})} = \theta$. This concludes the proof of \autoref{thm:main}.
  
  Let us then show how to choose a suitable substitution $f$. Let $N$ be a fixed positive and \emph{even} integer, and
  define the $N$-uniform substitution $\sigma\colon \{0, 1\}^* \to \{0, 1, \#\}^*$ by
  \begin{align*}
    0 &\mapsto \# 0^{N-1}, \\
    1 &\mapsto \# 1^{N-1}.
  \end{align*}

  \begin{lemma}\label{lem:trivial}
    The substitution $\sigma$ satisfies Property (i).
  \end{lemma}
  \begin{proof}
    Property (i) trivially holds for any nonerasing substitution.
  \end{proof}
	
  Before showing that the substitution $\sigma$ satisfies Property (ii), we show that the period of an abelian power
  with large enough exponent is divisible by $N$, the length of the substitution $\sigma$.

	\begin{lemma}\label{lem:div_length}
    Let $\infw{w}$ be an infinite binary word. If an abelian power $u_0 \cdots u_{e-1}$, with $e \geq N$, occurs in
    $\sigma(\infw{w})$, then $N$ divides $\abs{u_0}$.
	\end{lemma}
	\begin{proof}
    Let $m = \abs{u_0}$, and write $m = tN + r$ for some $t\geq 0$ and $0\leq r < N$. The claim is thus that $r=0$.
    Assume, for a contradiction, that $r>0$. Observe that for $\sigma(\infw{w}) = a_0a_1\cdots$, where
    $a_n\in \{0,1,\#\}$ for each $n\geq 0$, we have $a_n = \#$ if and only if $n\equiv 0\pmod{N}$. Let us denote the
    position of the occurrence of $u_j$ in $\sigma(\infw{w})$ by $i_j$, that is,
    \begin{equation*}
      u_j = a_{i_j}a_{i_j+1}\cdots a_{i_j+ m-1}.
    \end{equation*}
    Observe that $i_j = i_0 + j m$, and $i_j \equiv i_0+j r \pmod{N}$ for each $j=0$,$\ldots$,$e-1$. Notice also that
    the number of occurrences of the letter $\#$ in $u_j$ equals the number of indices $k$ in the set
    $\{i_j,i_{j}+1,\ldots, i_{j}+m-1\}$ for which $k\equiv 0 \pmod{N}$. Let $n_j = i_j \mod N$. If $n_j = 0$, then we
    may compute the value $\abs{u_{j}}_{\#}$ as follows:
    \begin{equation*}
      \abs{u_{j}}_{\#} = \left\lceil \frac{m}{N} \right\rceil = \left\lceil \frac{tN+r}{N} \right\rceil = t + \left\lceil \frac{r}{N} \right\rceil =t+1
    \end{equation*}
    since $0<r<N$ by assumption. If $n_j>0$, then none of the first $N-n_j$ letters of $u_j$ equals $\#$. The value
    $\abs{u_j}_{\#}$ is thus computed as follows:
    \begin{equation*}
      \abs{u_j}_{\#} = \left\lceil \frac{m-(N-n_j)}{N} \right\rceil = \left\lceil \frac{tN + r -(N-n_j)}{N} \right\rceil = t-1 + \left\lceil \frac{r + n_j}{N} \right\rceil.
    \end{equation*}
    We conclude that $\abs{u_j}_{\#} = t + 1$ if and only if $n_j=0$ or $n_j > N - r$, and otherwise
    $\abs{u_j}_{\#} = t$.
      
    We exhibit two words $u_{j_1}$ and $u_{j_2}$ from the abelian power for which the number of occurrences of the
    letter $\#$ differ. This contradiction proves our claim.	Since $e\geq N$, we see that the numbers $n_j$,
    $n_j \equiv n_0 + jr \pmod{N}$, $j=0$,$\ldots$,$e-1$, form the coset $n_0 + \langle r \rangle$ of the subgroup
    $\langle r \rangle$ of $\mathbb{Z}/{N\mathbb{Z}}$. Let now $d = \gcd(r,N)$, so that
    $\langle r \rangle = \{0,d,2d,\ldots,(N/d-1)d\}$. For example, if $\gcd(r,N) = 1$, then
    $\langle r\rangle = \mathbb{Z}/{N\mathbb{Z}}$. There thus exists an index $j_1$ such that the letter $\#$ occurs
    among the first $d$ letters of $u_{j_1}$. This means that either $n_{j_1}=0$ or
    \begin{equation*}
      n_{j_1} > N-d \geq N-r.
    \end{equation*} Thus $\abs{u_{j_1}}_{\#} = t+1$ as was concluded previously.
    Similarly, there exists an index $j_2$ such that the letter $\#$ occurs among the $d$ letters immediately preceding
    $u_{j_2}$. This means that 
    \begin{equation*}
      0< n_{j_2} \leq d.
    \end{equation*}
    In this case 
    \begin{equation*}
      n_{j_2} + r \leq d + r \leq d + N-d = N
    \end{equation*}
    since $r \leq N-\gcd(r,N) = N-d$. We thus have $n_{j_2} \leq N-r$ implying that $\abs{u_{j_2}}_{\#} = t$ as was
    concluded previously. This concludes the proof.
	\end{proof}
	
	\begin{remark}\label{rem:generalization_of_divisibility_property}
    The above result may be slightly generalized. Indeed, notice that the only structural properties of $\sigma$ used
    in the above proof are that $\sigma$ is uniform, the images of the letters begin with $\#$, and the images of the
    letters contain no other occurrences of $\#$. In fact, the property that both images of letters begin with $\#$ is
    not important, it is only required that $\#$ occurs at the same position in both $\sigma(0)$ and $\sigma(1)$. We
    are thus led to the following generalization of \autoref{lem:div_length}. Let
    $\varphi\colon \{0,1\}^*\to \{0,1,\#\}^*$ be a uniform substitution defined by $\varphi(0) = u\#v$,
    $\varphi(1) = u'\#v'$, where $u,u',v,v'\in\{0,1\}^*$, $\abs{u} = \abs{u'}$, and $\abs{v} = \abs{v'}$. Let
    $\infw{w}$ be a binary word. If an abelian power $u_0\cdots u_{e-1}$, $e\geq \abs{u\#v}$, occurs in
    $\varphi(\infw{w})$, then $\abs{u\#v}$ divides $\abs{u_0}$. We shall need this generalization later in
    \autoref{sec:k-abelian}.
	\end{remark}
	
	\begin{lemma}\label{lem:abelian_power_in_preimage}
    The substitution $\sigma$ satisfies Property (ii).
	\end{lemma}
	\begin{proof}
    Let $u_0 \cdots u_{e-1}$, $e \geq N$, be an abelian power occurring in $\sigma(\infw{w})$. It follows by
    \autoref{lem:div_length} that $N$ divides the length of $u_0$. Our aim is to show that the abelian power
    $u_0 \cdots u_{e-1}$ can be shifted (to the left or the right) to obtain another abelian power
    $u'_0 \cdots u'_{e-1}$ with $\abs{u'_0} = \abs{u_0}$ such that each $u'_i$ begins with the letter $\#$. Before
    doing so, let us show how the main claim follows from this. Because $\sigma$ is injective, as is readily verified,
    there exist unique factors $v_0$, $\ldots$, $v_{e-1}$ of $\infw{w}$ of length $\abs{u_0}/N$ such that
    $\sigma(v_i) = u'_i$ for $i = 0$, $\ldots$, $e - 1$. Notice that $v_0 \cdots v_{e-1}$ is a factor of $\infw{w}$.
    Clearly the words $v_i$ are abelian equivalent as $\abs{v_i}_0 = \abs{u'_i}_{0}/(N-1)$ and
    $\abs{u'_i}_0 = \abs{u'_j}_0$ for all $j$. We conclude that the word $v_0 \cdots v_{e-1}$ is an abelian power in
    $\infw{w}$.
	
    Let us again write $\sigma(\infw{w}) = a_0a_1\cdots$ with $a_n\in\{0,1,\#\}$ for each $n\geq 0$. Let $u_0$ have the
    position $i$ in $\sigma(\infw{w})$, and let $n = i \mod N$. If $n = 0$ then we are done since we may choose
    $u'_i=u_i$ in the above (recall that $N$ divides $\abs{u_0}$). Also, if $n=1$, each word $u_j$, $j=0$, $\ldots$,
    $e-1$, is immediately preceded by $\#$ in $\sigma(\infw{w})$ and, moreover, each of the words ends with $\#$. By
    setting $u_j' = \#u_j\#^{-1}$, we see that $\#u_0\cdots u_{e-1} = u_0'\cdots u_{e-1}'\#$ occurs in
    $\sigma(\infw{w})$, and clearly $u_j' \sim u_0$ for each $j=0$,$\ldots$,$e-1$. Thus $u_0'\cdots u_{e-1}'$ is an
    abelian power of the claimed form. Assume now that $n>1$. Without loss of generality, we assume that $u_0$ begins
    with $0$ so, in fact, $u_0$ begins with $0^{N-{n}}\#$. By the form of the substitution, $u_0$ is preceded by
    $\#0^{n-1}$ in $\sigma(\infw{w})$. We claim that each of the words $u_j$, $j=0$, $\ldots$, $e-1$, begins with
    $0^{N-{n}}\#$ and ends with $\#0^{n-1}$. Let us first show that $u_1$ begins with $0^{N-n}\#$ (and thus that $u_0$
    ends with $\#0^{n-1}$). Assume for a contradiction that $u_1$ begins with $1^{N-n}\#$ (whence $u_0$ ends with
    $\#1^{n-1}$), and say that $u_1$ ends with $\#c^{n-1}$ where $c\in \{0,1\}$. Now the word
    $\#0^{n-1}u_0(\#1^{n-1})^{-1}$ is the image of a factor $x$ of $\infw{w}$. Similarly, the word
    $\#1^{n-1}u_1(\#c^{n-1})^{-1}$ is the image of a factor $y$ of $\infw{w}$ with $\abs{x} = \abs{y}$. We may write
    \begin{equation*}
      \abs{u_0}_{1} = \abs{x}_{1}(N-1) + n-1
    \end{equation*}
    and
    \begin{equation*}
      \abs{u_1}_{1} = \abs{y}_{1}(N-1) - (n-1) + \delta_{c=1}\cdot (n-1),
    \end{equation*}
    where $\delta_{c=1} = 1$ if $c=1$, and otherwise $\delta_{c=1} = 0$. Since $u_0 \sim u_1$, by rearranging the
    terms, we obtain 
    \begin{equation*}
      (\abs{y}_{1}-\abs{x}_{1})(N-1) = (2 - \delta_{c=1})(n-1).
    \end{equation*}
    Notice here that $1 \leq 2-\delta_{c=1} \leq 2$ and that $n>1$. The right side of the inequality is positive, so
    $\abs{y}_1 - \abs{x}_1 \geq 1$. Since $N > n$, it must be that $\abs{y}_1 - \abs{x}_1 < 2 - \delta_{c=1} \leq 2$.
    We conclude that $\abs{y}_1 - \abs{x}_1 = 1$ and, furthermore, $\delta_{c=1} = 0$. We now have
    \begin{equation*}
      N-1 = 2(n-1),
    \end{equation*}
    which is impossible since $N$ was chosen to be even. This contradiction shows that $u_1$ begins with $0^{N-n}\#$ as
    well. A symmetric argument shows that $u_1$ ends with $\#0^{n-1}$. We may repeat the above argument to show that
    each of the words $u_j$, $j=0$,$\ldots$,$e-1$, begins with $0^{N-n}\#$ and ends with $\#0^{n-1}$.
	
    To finish off the proof, we choose $u_j' = \#0^{n-1}u_j(\#0^{n-1})^{-1}$ for each $j=0$,$\ldots$,$e-1$. Observe that
    $\#0^{n-1}u_0\cdots u_{e-1} = u_0'\cdots u_{e-1}'\#0^{n-1}$ and that $u_0' \sim u_j'$ for each
    $j=0$,$\ldots$,$e-1$. We have thus exhibited an abelian power of the claimed form thus concluding the proof.
	\end{proof}
	
  Since the substitution $\sigma$ satisfies Properties (i)-(ii) and $\sigma(\infw{w})$ has at most three letters,
  \autoref{thm:main} is proved.
	
  \section{Extension to the \texorpdfstring{$k$}{k}-abelian Setting}\label{sec:k-abelian}
  In this section, we consider a generalization of abelian equivalence. Let $k$ be a positive integer. Two words $u$
  and $v$ are \emph{$k$-abelian equivalent}, written $u \sim_k v$, if $\abs{u}_w = \abs{v}_w$ for all nonempty words
  $w$ of length at most $k$ \cite{2013:on_a_generalization_of_abelian_equivalence_and_complexity_of_infinite}. For
  words of length at least $k - 1$, we can equivalently say that $u \sim_k v$ if and only if $u$ and $v$ share a common
  prefix and a common suffix of length $k - 1$ and $\abs{u}_w = \abs{v}_w$ for each word $w$ of length $k$
  \cite[Lemma~2.4]{2013:on_a_generalization_of_abelian_equivalence_and_complexity_of_infinite}. The $k$-abelian
  equivalence relation is a congruence relation. Notice that $1$-abelian equivalence is simply abelian equivalence.
  Moreover, if $u \sim_{k+1} v$, then $u \sim_k v$.

  A nonempty word $u_0 \cdots u_{e-1}$ is a \emph{$k$-abelian power} of exponent $e$ and period $m$ if $\abs{u_0} = m$
  and $u_0 \sim_k \cdots \sim_k u_{e-1}$. It was proved in
  \cite[Thm.~5.4]{2013:on_a_generalization_of_abelian_equivalence_and_complexity_of_infinite} using Szemer\'{e}di's
  theorem that every infinite word having bounded $k$-abelian complexity contains $k$-abelian powers of arbitrarily
  high exponent. Sturmian words are particular examples of such words, so each Sturmian word contains $k$-abelian
  powers of arbitrarily high exponent; an alternative proof of this fact is given in
  \cite[Lemma~3.10]{2018:on_k-abelian_equivalence_and_generalized_lagrange_spectra}
  
  Let $\infw{w}$ be an infinite word. Then we set $\abexp{k,\infw{w}}{m}$ to be the supremum of the exponents of
  $k$-abelian powers of period $m$ occurring in $\infw{w}$. We define the \emph{$k$-abelian critical exponent} of
  $\infw{w}$ to be the quantity
  \begin{equation*}
    \limsup_{m \to \infty} \frac{\abexp{k,\infw{w}}{m}}{m},
  \end{equation*}
  and we denote it by $\act[k]{\infw{w}}$. This generalization of the abelian critical exponent is considered in the
  preprint \cite{2018:on_k-abelian_equivalence_and_generalized_lagrange_spectra}, where the authors of this paper study
  the set of finite $k$-abelian critical exponents of Sturmian words. This set, dubbed as the \emph{$k$-Lagrange
  spectrum}, is similarly complicated as the Lagrange spectrum. When $k > 1$, the least accumulation point of the
  $k$-Lagrange spectrum is $\sqrt{5}/(2k - 1)$, and the spectrum is dense in the interval
  $(\sqrt{5}/(2k - 1), \infty)$.

  Next we prove the following analogue of \autoref{thm:main}.

  \begin{theorem}\label{thm:main_2}
    Let $\theta$ be a nonnegative real number. Then there exists an infinite word $\infw{w}$ such that
    $\act[k]{\infw{w}} = \theta$. The word $\infw{w}$ can be taken over an alphabet of at most three letters.
  \end{theorem}

  Similar to \autoref{sec:abelian}, we wish to find a substitution $f$ defined on a two-letter alphabet with the
  following properties:
  \begin{enumerate}[(i')]
    \item If an abelian power $u_0 \cdots u_{e-1}$ occurs in $\infw{w}$, then $f(u_0) \cdots f(u_{e-1})$ is a
          $k$-abelian power occurring in $f(\infw{w})$.
    \item If a $k$-abelian power $u_0 \cdots u_{e-1}$, $e \geq N$, occurs in $f(\infw{w})$, then $\infw{w}$ contains an
          abelian power $v_0 \cdots v_{e-1}$ with $\abs{v_0} = \abs{u_0}/N$.
  \end{enumerate}

  Given such a substitution $f$, \autoref{thm:main_2} is proved exactly as \autoref{thm:main} was proved in
  \autoref{sec:abelian}. The case $k = 1$ is handled by \autoref{thm:main}, so we may assume that $k > 1$.

	Let $N \geq 2k-1$ be a fixed integer, and define the $N$-uniform substitution
  $\tau\colon \{0,1\}^* \to \{0,1,\#\}^*$ by
	\begin{align*}
    0 &\mapsto \# 0^{k-2} 0^{N-2k+2} 0^{k-1}, \\
    1 &\mapsto \# 0^{k-2} 1^{N-2k+2} 0^{k-1}.
  \end{align*}

  Let $u$ and $v$ be two words of length greater than $2k - 2$. Suppose that $\pref_{k-1}(u) = \pref_{k-1}(v)$ and
  $\suff_{k-1}(u) = \suff_{k-1}(v)$, that is, assume that they share a common prefix of length $k - 1$ and a common
  suffix of length $k - 1$. One easily checks that then $uv \sim_k vu$. Remark then that it follows that
  $xuvy \sim_k xvuy$ for all words $x$ and $y$ because $\sim_k$ is a congruence.

  \begin{lemma}
    The substitution $\tau$ satisfies Property (i').
  \end{lemma}
  \begin{proof}
    By the form of the substitution $\tau$, we have $\pref_{k-1}(\tau(0)) = \pref_{k-1}(\tau(1))$ and
    $\suff_{k-1}(\tau(0)) = \suff_{k-1}(\tau(1))$. Therefore $\tau(0)\tau(1) \sim_k \tau(1)\tau(0)$, and hence
    $\tau(u_i) \sim_k \tau(0)^{\abs{u_i}_0} \tau(1)^{\abs{u_i}_1}$ for $i = 0$, $\ldots$, $e - 1$. Let
    $u_0 \cdots u_{e-1}$ be an abelian power in $\infw{w}$. Since the words $u_0$, $\ldots$, $u_{e-1}$ are abelian
    equivalent, we have
    \begin{equation*}
		\tau(u_0), \ldots, \tau(u_{e-1}) \sim_k \tau(0)^{\abs{u_0}_0} \tau(1)^{\abs{u_0}_1},
		\end{equation*} so
    $\tau(u_0) \sim_k \cdots \sim_k \tau(u_{e-1})$.
  \end{proof}

  \begin{lemma}\label{lem:div_k}
    If $u_0 \cdots u_{e-1}$, $e\geq N$, is a $k$-abelian power occurring in $\tau(\infw{w})$,
		then $N$ divides $\abs{u_0}$.
  \end{lemma}
  \begin{proof}
	Since $u_0\cdots u_{e-1}$ is a $k$-abelian power, it is an abelian power. Observe now that the substitution
	$\tau$ is as in \autoref{rem:generalization_of_divisibility_property}. Thus $N$ divides $\abs{u_0}$.
  \end{proof}

  \begin{lemma}
    The substitution $\tau$ satisfies Property (ii').
  \end{lemma}
  \begin{proof}
    Suppose that a $k$-abelian power $u_0 \cdots u_{e-1}$ with $e \geq N$ occurs in $\tau(\infw{s})$. By
    \autoref{lem:div_k}, $N$ divides $\abs{u_0}$. Similar to the proof of \autoref{lem:abelian_power_in_preimage}, we
    want to show that the $k$-abelian power $u_0 \cdots u_{e-1}$ can be shifted (to the left or the right) to obtain
    another $k$-abelian power $u'_0 \cdots u'_{e-1}$, $\abs{u'_0} = \abs{u_0}$, such that each $u'_i$ begins with $\#$.
    Then a slight modification of the argument presented in the first paragraph of the proof of
    \autoref{lem:abelian_power_in_preimage} proves the claim. Indeed, given the preimages $v_0$, $\ldots$, $v_{e-1}$ of
    $u'_0$, $\ldots$, $u'_{e-1}$, we see that $\abs{v_i}_0 = \abs{u'_i}_{\#0^{k-1}}$ for all $i$. Since
    $u'_0 \sim_k \cdots \sim_k u'_{e-1}$, we have $\abs{u'_i}_{\#0^{k-1}} = \abs{u'_j}_{\#0^{k-1}}$ for all $i$ and
    $j$, and it follows that $v_0 \sim \cdots \sim v_{e-1}$.

    Let $p$ be the common prefix of length $k - 1$ of the words $u_0$, $\ldots$, $u_{e-1}$ and similarly $q$ be the
    common suffix of length $k - 1$ of these words. Suppose first that $\#$ occurs in $p$, that is, $p = 0^r \# 0^s$
    with $r + s = k - 2$. As each occurrence of $\#$ is preceded by $0^{k-1}$ and $N$ divides $\abs{u_i}$, the word
    $u_{e-1}$ is followed by $0^r$. Thus we may set $u'_i = (0^r)^{-1} u_i 0^r$ for $i = 0$, $\ldots$, $e - 1$. The
    same $r$ factors $0^r \# 0^{s+1}$, $0^{r-1} \# 0^{s+2}$, $\ldots$, $0 \# 0^{s+r}$ of length $k$ were removed from
    each $u_i$ and the same $r$ factors of length $k$ were added to each $u'_i$ (the final $k-1$ factors of $q 0^r$ of
    length $k$) during the shift. Thus $u'_0 \sim_k \cdots \sim_k u'_{e-1}$. If $\#$ occurs in $q$, that is, say
    $q = 0^r \# 0^s$ with $r + s = k - 2$ then, like above, we may set $u'_i = \# 0^s u_i (\# 0^s)^{-1}$ for
    $i = 0$, $\ldots$, $e - 1$. Suppose then that some word $u_i$ begins with $0^{k-1} \#$. It is straightforward to
    see that then all of the words $u_0$, $\ldots$, $u_{e-1}$ begin with $0^{k-1} \#$ and, furthermore, that $u_{e-1}$
    is followed by $0^{k-1} \#$. Setting $u'_i = (0^{k-1})^{-1} u_i 0^{k-1}$ for $i = 0$, $\ldots$, $e - 1$ gives the
    claim as above.

    By the preceding paragraph, we may assume that the occurrence of $p$ as the prefix of $u_i$ is a proper factor of
    $\tau(c_i)$ for a letter $c_i$, that is, we may write $\tau(c_i) = x_i p y_i$, with $x_i$ and $y_i$ nonempty, for
    this occurrence of $p$. Moreover, the preceding paragraph tells that we may assume that $q$ is a proper suffix of
    $x_i$ (otherwise $\#$ occurs in $q$). Since $p$ has length $k - 1$, it is clear from the form of the substitution
    $\tau$ that the letter $c_i$ is uniquely determined by $p$. Since $N$ divides $\abs{u_i}$, it follows that
    $c_0 = \ldots = c_{e-1}$. This means that each $u_i$ is preceded by $x_0$. We still need to know that $u_{e-1}$
    ends with $x_0$; the words $u_0$, $\ldots$, $u_{e-2}$ must end with $x_0$. Since $N$ divides $\abs{u_i}$, the
    suffix $q$ of $u_{e-1}$ occurs in $\tau(d)$, $d \in \{0, 1\}$, in the same position as the occurrence of $q$
    preceding $py_0$ in $\tau(c_0)$. Now $q$ has length $k - 1$, so its occurrence preceding $py_i$ in $\tau(c_i)$
    uniquely determines $c_i$, and hence its occurrence in $\tau(d)$ in the same position uniquely determines $d$.
    Therefore $d = c_0$ and $u_{e-1}$ ends with $x_0$. We may now set $u'_i = x_0 u_i x_0^{-1}$ for $i = 0$, $\ldots$,
    $e - 1$. The suffix $x_0$ of $u_i$ is preceded by $0^{k-1}$ and $u_i$ has prefix $p$ of length $k - 1$, so exactly
    the same factors of length $k-1$ are added and removed when shifting each $u_i$ to $u'_i$. Thus
    $u'_0 \sim_k \cdots \sim_k u'_{e-1}$.
  \end{proof}

  Since $\tau$ satisfies Properties (i') and (ii'), \autoref{thm:main_2} follows.

  \section{Concluding Remarks}
  \autoref{thm:main} raises the following question.

  \begin{question}
    Given a nonnegative real number $\theta$, does there exist an infinite binary word having $k$-abelian critical
    exponent $\theta$?
  \end{question}

  We conjecture that the question has a positive answer. To use the presented method, the marker letter $\#$ needs to
  be replaced by a suitable binary word ensuring that Properties (ii) and (ii') hold. There seems to be no obvious
  choice, at least no obvious choice leading to reasonable proofs. Perhaps another method is required. It would
  certainly be very interesting if the answer to the above question turned out to be negative. Nevertheless, we leave
  the question open.

  The $k$-abelian equivalence is a refinement of abelian equivalence that ``tends'' to the usual equality of words as
  $k \to \infty$. As mentioned in the introduction, it is typical to consider the maximum exponent
  $\sup_{m \geq 1} \exp(m)$ for the equality relation, not the superior limit of the ratio between the maximum exponent
  $\exp(m)$ and period $m$ as is done here for abelian equivalence and $k$-abelian equivalence. What then happens if we
  consider the unorthodox notion? Does an analogue to \autoref{thm:main} hold? The answer is yes. The following result
  is proved by the authors in the preprint \cite{2018:on_k-abelian_equivalence_and_generalized_lagrange_spectra}.

  \begin{proposition}\cite[Prop.~3.17]{2018:on_k-abelian_equivalence_and_generalized_lagrange_spectra}
    Given an infinite word $\infw{w}$, let $E(\infw{w})$ be the quantity
    \begin{equation*}
      \limsup_{m \to \infty} \frac{\exp(m)}{m},
    \end{equation*}
    where $\exp(m)$ is the supremum of (integral) exponents of powers of period $m$ occurring in $\infw{w}$. For each
    nonnegative $\theta$, there exists a Sturmian word $\infw{s}$ such that $E(\infw{s}) = \theta$.
  \end{proposition}

  \bibliographystyle{splncs04}
  \bibliography{references}
  
\end{document}